\documentclass[12pt]{iopart}
\usepackage{iopams}
\usepackage{graphicx}
\usepackage{float}
\usepackage{mdframed}
\expandafter\let\csname equation*\endcsname\relax
\expandafter\let\csname endequation*\endcsname\relax
\usepackage{amsmath}
\usepackage{mathtools} 
\usepackage{braket} 
\usepackage{stackrel} 
\usepackage{bigints} 
\usepackage{bbm} 
\usepackage{amsthm} 
\usepackage{pifont}
\newtheorem{remark}{Remark}
\newtheorem{theorem}{Theorem}[section]

\newtheorem{lemma}[theorem]{Lemma}
\theoremstyle{definition}
\newtheorem{definition}{Definition}[section]

\usepackage[usernames,dvipsnames,svgnames,table]{xcolor}
\usepackage{hyperref}
\definecolor{myblue}{rgb}{0.0, 0.3, 1}
\newcommand{\mini}{\scriptscriptstyle}

\DeclareMathOperator{\Imag}{Im}

\begin{document}

	\title[Hidden freedom in the mode expansion on static spacetimes]{Hidden freedom in the mode expansion \\on static spacetimes}

	\author{Lissa de Souza Campos$^1$, Claudio Dappiaggi$^2$, Luca Sinibaldi$^3$}
	\address{Dipartimento di Fisica, Universit\`a degli Studi di Pavia \& \\
		Istituto Nazionale di Fisica Nucleare, Sezione di Pavia, Via Bassi, 6, Pavia, Italy}
	\ead{$^1$lissa.desouzacampos@unipv.it\\\hspace{1.3cm}$^2$claudio.dappiaggi@unipv.it\\\hspace{1.3cm}$^3$luca.sinibaldi01@universitadipavia.it}
	\vspace{10pt}
	\begin{indented}
		\item[]July 2022
	\end{indented}

	\begin{abstract}
We review the construction of ground states focusing on a real scalar field whose dynamics is ruled by the Klein-Gordon equation on a large class of static spacetimes. As in the analysis of the classical equations of motion, when enough isometries are present, via a mode expansion the construction of two-point correlation functions boils down to solving a second order, ordinary differential equation on an interval of the real line. Using the language of Sturm-Liouville theory, most compelling is the scenario when one endpoint of such interval is classified as a limit circle, as it often happens when one is working on globally hyperbolic spacetimes with a timelike boundary. In this case, beyond initial data, one needs to specify a boundary condition both to have a well-defined classical dynamics and to select a corresponding ground state. Here, we take into account boundary conditions of Robin type by using well-known results from Sturm-Liouville theory, but we go beyond the existing literature by exploring an unnoticed freedom that emerges from the intrinsic arbitrariness of secondary solutions at a limit circle endpoint. Accordingly, we show that infinitely many one-parameter families of sensible dynamics are admissible. In other words, we emphasize that physical constraints guaranteeing the construction of full-fledged ground states do not, in general, fix one such state unambiguously. In addition, we provide, in full detail, an example on $(1 + 1)$-half Minkowski spacetime to spell out the rationale in a specific scenario where analytic formulae can be obtained.


	\end{abstract}
	\maketitle

	\section{Introduction}

	Quantum field theory on curved spacetimes has lead to significant improvements in our understanding of different physical phenomena ranging from particle production in cosmology to Hawking radiation in black hole physics. In the analysis of the vast majority of the available models, the first step consists in constructing full-fledged quantum states for free fields. Under the mild assumption that the correlation functions are Gaussian, it reduces to the identification of an on-shell two-point correlation function that has to abide to physically motivated constraints. The prime example in this direction is the Hadamard condition, which ensures not only that the quantum fluctuations of all observables are finite, but also that Wick ordered fields can be constructed following a locally covariant scheme. In turn, it entails control of the underlying renormalization group and of interactions that are studied at a perturbative level. Yet, in many concrete scenarios one is limited to abstractly argue the existence of such distinguished two-point functions and an explicit construction is, at best, elusive.

Major improvements occur when one concentrates on static spacetimes $\mathcal{M} \simeq \mathbb{R}\times \Sigma$, regardless of whether they are globally hyperbolic or not. Let us consider, for simplicity, a free, scalar field $\Psi:\mathcal{M}\rightarrow \mathbb{R}$ that abides to the Klein-Gordon equation. By calling $t$ the time coordinate along $\mathbb{R}$, the latter simplifies to
	\begin{equation}\label{eq: KG evolution form}
		\left(\frac{\partial^2}{\partial t^2} + K \right)\Psi = 0,
	\end{equation}
	where $K$ is an elliptic, second order partial differential operator. The key rationale consists of reading $K$ as a symmetric operator on the Hilbert space $\mathcal{H}:= L^2(\Sigma,d\mu_\Sigma)$ of square-integrable functions with respect to the measure $d\mu_\Sigma$ induced by the Lorentzian metric tensor of $\mathcal{M}$ on $\Sigma$. This leads to two notable advantages, one at a classical and one at a quantum level, as described next.

At a classical level, solutions of \Eref{eq: KG evolution form} can be constructed as follows. Since $K$ is a real and symmetric operator, it admits a non-necessarily unique self-adjoint extension $\widetilde{K}$. Assuming, for convenience, that $\widetilde{K}$ has positive spectrum, and given initial data $(\Psi_0,\dot\Psi_0)\in C_0^\infty(\Sigma)\times C_0^\infty(\Sigma)\cap D(\widetilde{K})\times D(\widetilde{K})$, for each $t\in\mathbb{R}$ we have
	\begin{equation}
		\label{eq: sol Psi_t dynamics A_E}
		\Psi_t:=\cos(\sqrt{\widetilde{K}}t)\Psi_0+ \frac{\sin(\sqrt{\widetilde{K}}t)}{\sqrt{\widetilde{K}}}\dot\Psi_0\in\mathcal{H},
	\end{equation}
where each term is well-defined using spectral calculus. Moreover, there exists a unique $\Psi\in C^\infty(\mathcal{M})$ such that,
	\begin{equation*}
		\Psi|_{\mini{\Sigma_t}} = \Psi_t   \quad \text{ and }\quad  \nabla_n\Psi|_{\mini{\Sigma_t}} = \dot\Psi_t,
	\end{equation*}
where $\Sigma_t\equiv\{t\}\times\Sigma$, $t\in\mathbb{R}$, while $n$ is the unit vector field normal to $\Sigma_t$.

If $(\mathcal{M},g)$ is a globally hyperbolic spacetime without boundary, than there exists a unique choice for $\widetilde{K}$ and the dynamics is therefore unambiguously determined \cite{bar2007wave}. On the contrary, if $K$ has more than one self-adjoint extension, then multiple, physically inequivalent scenarios do exist. Markedly, the latter is not a remote possibility: it occurs for example when $(\mathcal{M},g)$ is a globally hyperbolic spacetime with a timelike boundary, see \cite{Ake2018dzz}. This class of backgrounds encompasses several physically interesting scenarios such as AdS spacetime, which was first analyzed in the language above by Ishibashi and Wald in \cite{Ishibashi2003jd}. It must be stressed that, using the language of boundary triples \cite{grubb1968characterization}, the infinite set of different choices of self-adjoint extensions for $\widetilde{K}$ can be put in correspondence with the choice of a boundary condition for \Eref{eq: KG evolution form}.

At a quantum level, the assumption that $(\mathcal{M},g)$ is static guarantees a considerable advantage: the existence of a {\em ground state}. Under the same premises of the previous paragraphs, the associated two-point function $\psi_2(t,x,t^\prime, x^\prime)$ can be constructed directly as the integral kernel of the operator \cite{fulling1989aspects}
\begin{equation}\label{eq: ground state static}
\frac{e^{-i\sqrt{\widetilde{K}}(t-t^\prime)}}{2\sqrt{\widetilde{K}}}.
\end{equation}
Observe that there exists a different ground state for each self-adjoint extension $\widetilde{K}$, in accordance with the fact that each $\widetilde{K}$ characterizes a different physical system. Furthermore, if the underlying background is a globally hyperbolic spacetime with or without timelike boundary, then ground states are of Hadamard form as a consequence of the results of \cite{Sahlmann:2000fh,Dappiaggi2021wtr}.

While at this stage the analysis of a scalar quantum field theory on a static spacetime seems a rather well-understood problem, the drawback lies in two crucial details. On the one hand, when non unique, an explicit construction and characterization of all self-adjoint extensions of $K$ in Equation \eqref{eq: KG evolution form} is a daunting task. On the other hand, the quantitative evaluation of physical observables in concrete scenarios, such as on black hole spacetimes, requires a deeper and more hands on knowledge of the two-point function, far beyond the spectral level as per Equation \eqref{eq: ground state static}. To bypass this conceptual hurdle, it is customary to consider static backgrounds with a high degree of symmetry. Beyond reasons of mathematical simplicity, this class includes many physically relevant backgrounds, such as cosmic strings, black holes and asymptotically AdS spacetimes.

In this paper, we consider the class of $n$-dimensional static spacetimes $\mathcal{M}$ that are isometric either to
\begin{equation*}
\mathbb{R}\times I, \,\text{  if  }\, n=2 \quad \text{  or  }\quad \mathbb{R}\times I\times\Sigma_j^{n-2} , \,\text{  if  }\, n>2,
\end{equation*}
where $I\subseteq\mathbb{R}$, and $\Sigma_j^{n-2}$ are Cauchy-complete, connected, $(n-2)$-dimensional Riemannian manifolds of constant sectional curvature $j$. The line element associated to the metric tensor on $\mathcal{M}$ reads
 \begin{equation*}
	ds^2=-f(r)dt^2+h(r)dr^2+r^2d\Sigma_{j}^{n-2}(\varphi_1,\dots,\varphi_{n-2}),
\end{equation*}
where $f$ and $h$ are suitable positive functions. Barring some technical aspects that will be specified in the next sections, we emphasize that a large class of spacetimes is characterized by the line-element above, including black hole backgrounds ranging from the three-dimensional static BTZ spacetime to the $n$-dimensional Schwarzschild or Schwarzschild-AdS spacetime.

On top of these manifolds, we consider a real, scalar field $\Psi$ whose dynamics is ruled by the Klein-Gordon equation, which, as before, can be written as per \Eref{eq: KG evolution form}. With the construction of a quantum field theoretical framework in mind, we are interested in obtaining distinguished two-point functions that correspond to full-fledged ground states. Although the procedure outlined above is applicable, especially when considering scenarios where boundary conditions needs to be imposed, it is common to follow a more computationally oriented approach that exploits the underlying symmetries. In the following, we sketch the steps usually followed in the literature in these scenarios. More details will be given in the next sections of this work.\label{pg:sketch items 1}
\begin{itemize}
	\item[\ding{192}] Consider a solution of the Klein-Gordon equation on $(\mathcal{M},g)$ assuming that it admits a mode expansion:
	\begin{equation*}
		\Psi_{\omega \eta_j}(t,r,\varphi_1,...,\varphi_{n-2})= e^{-i\omega t}R_{\omega \eta_j}(r)Y_j(\varphi_1,...,\varphi_{n-2}),
	\end{equation*}
	where $Y_j(\varphi_1,...,\varphi_{n-2})$ are the eigenfunctions of the Laplace operator on $\Sigma_{j}^{n-2}$ whose corresponding eigenvalue is denoted by $\eta_j$, while $\omega\in\mathbb{R}$ plays the standard r\^{o}le of frequency.
	\item[\ding{193}] The only unknown function $R_{\omega \eta_j}$ can be shown to satisfy an eigenvalue equation $\mathbf{A}R_{\omega \eta_j}=\lambda R_{\omega \eta_j}$ where $\mathbf{A}$ is a second order differential operator in the radial coordinate $r$ whose domain is the interval $I$, here taken for definiteness as $(a,b)$. Most notably $\mathbf{A}$ can be written in the form of a, possibly singular, Sturm-Liouville operator, see \cite{zettl}.
	\item[\ding{194}] Similarly to the r\^{o}le played by $K$ in Equation \eqref{eq: KG evolution form}, one reads $\mathbf{A}$ as a symmetric operator on a space of square-integrable functions over the interval $I=(a,b)$.
	\item[\ding{195}] Following von Neumann's theory of deficiency indices \cite{Moretti}, three options are possible. $\mathbf{A}$ can admit just one self-adjoint extension, a one-parameter or a two-parameters family of self-adjoint extensions. In most applications, the last option does not occur.
\end{itemize}

At this stage it is necessary to pause the description of the procedure to construct a two-point function and draw the attention to self-adjoint extensions. More precisely, from the viewpoint of the differential equation $\mathbf{A}R_{\omega \eta_j}=\lambda R_{\omega \eta_j}$, the existence of such extensions can be inferred by looking at the behavior of solutions close to $a$ and $b$, the endpoints of the interval $I$. Henceforth, for definiteness, we focus on $a$. The general theory of Sturm-Liouville operators guarantees that, for any $\lambda\in\mathbb{C}$, there is always a distinguished function $u$, called {\em principal solution}. In Section \ref{sec: Self-adjoint extensions} we dwell on the technical details of this concept. For now it suffices to say that $u$ tends to zero as $r\to a^+$ faster than any other solution that is linearly independent from it, and it is square-integrable in any neighborhood of the endpoint $a$. Notwithstanding, the existence of another solution, called {\em secondary solution}, that is linearly independent from $u$ and square-integrable in any neighborhood of $a$ depends on the differential problem at hand. If such a solution does not exist at both endpoints, then it happens that there is a unique self-adjoint extension for $\mathbf{A}$, see \cite{zettl}. More interesting is the scenario for which that is not the case for some value of $\lambda\in\mathbb{C}$, on account of the fact that if a secondary solution exists, then it is highly non unique.

The intrinsic arbitrariness of the secondary solution lies at the core of this work. Suppose there exists a secondary solution at $a$, but only the principal one at $b$. From the viewpoint of the Sturm-Liouville operator, the choice of a specific secondary solution at $a$ is irrelevant as its r\^{o}le lies only in establishing a one-to-one correspondence between self-adjoint extensions of $\mathbf{A}$ and boundary conditions of Robin type, assigned at the endpoint $a$. These exhaust all possibilities at the level of the ordinary differential equation and therefore one can read the choice of two different secondary solutions as two different, albeit equivalent, ways to span the same space of boundary conditions and, consequently, of solutions of the underlying ordinary differential equation. In most of the recent literature in quantum field theory on curved spacetimes, a consequential dogma has always been to consider the secondary solution as an innocuous and physically insignificant abstraction. Our main goal is to argue that this by far not the case since, tracking back the problem at a fully covariant level, making different choices of the secondary solution, while keeping Robin boundary conditions, impacts significantly the analysis of scalar field $\Psi$. More precisely one is able to codify at a covariant level a much larger class of admissible boundary conditions than just a one-parameter family as the standard analysis might suggest.

The physical relevance of choosing a secondary solution is limpid when one constructs the two-point correlation function of the underlying ground state. Let us thus focus once more on the procedure we started sketching at Page \pageref{pg:sketch items 1} and let us state the subsequent steps, as follows.
\begin{itemize}
	\item[\ding{196}] Since the two-point correlation function $\psi_2$ must obey the Klein-Gordon equation it is suitable to use the same mode expansion as for the construction of the solutions $\Psi$. This, together with the ansatz that only positive frequencies contribution are of relevance, identifies a full-fledged ground state. On account of the large isometry group of the background, $\psi_2$ is completely determined up to a kernel along the radial direction.
	\item[\ding{197}] As explained in Section \ref{sec: The quantum dynamics}, such kernel can be constructed using an algorithmic scheme once a self-adjoint extension for the operator $\mathbf{A}$ has been chosen in terms of a Robin boundary conditions imposed at the level of principal and secondary solutions.
\end{itemize}

We emphasize that the procedure \ding{192}--\ding{197} has been extensively applied on static spacetimes with a timelike boundary in the last years. To mention a few, exhaustive works based on Von Neumann deficiency index theory are \cite{Higuchi2021fxg,ishibashi2004dynamics}, on AdS spacetimes, and \cite{Garbarz2017wzv}, on a static BTZ black hole. Beyond these, analyses based on a mode expansion and on Robin boundary conditions and used in the construction of physically-sensible two-point functions within quantum field theory on asymptotically AdS spacetimes can be found in \cite{Barroso2019cwp,Barroso2018pjs,Dappiaggi2018xvw,Bussola2017wki,Dappiaggi2016fwc}, and, more recently, also in \cite{Dappiaggi2022dwo,Campos2021wyr,deSouzaCampos2020bnj,Morley2020zcd,Morley2020ayr}.

The scope of this work is to highlight the fundamental r\^{o}le played, in this whole procedure, by the secondary solutions of the underlying Sturm-Liouville problem. To strengthen this statement we consider a simple, yet illustrative example, namely the two-dimensional half-Minkowksi spacetime $\mathbb{R}\times\mathbb{R}_+$. The advantage is that, in this scenario, we can address the problem using explicit, analytic formulae that allow to make clear that even minor adjustments to the choice of secondary solution yield, at the fully covariant level, boundary conditions that have a completely different physical interpretation. In a nutshell, we aim to convey that the choice of secondary solution is of physical consequence.

	This paper is organized as follows. In Section \ref{sec: The Klein-Gordon equation} we show that the Klein-Gordon equation untangles into a Sturm-Liouville problem for the radial part of the Klein-Gordon operator on static spacetimes with maximally symmetric sections. Subsequently, in Section \ref{sec: Self-adjoint extensions} we provide straightforward generalizations of main results from singular Sturm-Liouville theory that allow us to obtain all self-adjoint representations of the latter. Markedly, singular endpoints give rise to an ambiguity in the definition itself of generalized Robin boundary conditions. We clear up this ambiguity in Section \ref{sec: generalized gamma v robin} by defining generalized $(\gamma, v)-$Robin boundary conditions and explaining its connection with the regular case. In Section \ref{sec: Sensible dynamics} we show how a boundary condition on the radial part translates to a boundary condition on the full solution $\Psi$. Insofar as solutions of the Klein-Gordon equation characterize a classical dynamics, it is meaningful to include a discussion on the canonical quantization procedure. Hence, in Section \ref{sec: The quantum dynamics} we explain the connection between the imposition of the canonical commutation relations and the spectral resolution of the identity given by the before-mentioned Sturm-Liouville problem.
	We illustrate the main points of this work in a detailed example given in Section \ref{sec: example harmonic oscillator}. Most importantly, this example clarifies in which sense the generalized $(\gamma, v)-$Robin boundary conditions imposed on the radial part may render time-dependent boundary conditions on $\Psi$. Final remarks are given in Section \ref{sec: conclusions}.

	\section{The Klein-Gordon equation
		\label{sec: The Klein-Gordon equation}
	}

In this initial section we introduce both the geometric data of the spacetimes we are interested in and the Klein-Gordon equation. In addition we show that, under the specific assumptions on the background metric, the Klein-Gordon equation can be reduced to a Sturm-Liouville problem.

In this work, for $n>2$, $\Sigma_{j}^{n-2}$ denotes a Cauchy-complete, connected, $(n-2)$-dimensional Riemannian manifold of constant sectional curvature $j$, parametrized by $(\varphi_1,...,\varphi_{n-2})$, and whose standard metric has an associated line element $d\Sigma_j^{n-2}(\varphi_1,\dots,\varphi_{n-2})$.  Unless stated otherwise, we shall assume $j$ has been normalized so that $j\in\{-1,0,+1\}$. The symbol $\mathcal{M}$ refers to an $n$-dimensional, static spacetime isometric to the warped geometry $\mathbb{R}\times \text{I} \times \Sigma_{j}^{n-2}$, where $I\subseteq\mathbb{R}$, while the line element of the underlying metric read in global Schwarzschild-like coordinates $(t,r,\varphi_1,...,\varphi_{n-2})$:
	\begin{equation}
		\label{eq:metric schwarzschild-like}
		ds^2=-f(r)dt^2+h(r)dr^2+r^2d\Sigma_{j}^{n-2}(\varphi_1,\dots,\varphi_{n-2}).
	\end{equation}
For convenience, we shall call $t\in\mathbb{R}$, $r\in\text{I}\subseteq\mathbb{R}$, and $(\varphi_1,...,\varphi_{n-2})$, respectively, the time, the radial and the angular coordinates. Equation \eqref{eq:metric schwarzschild-like} is completely specified aside from the two functions $f,h$. For simplicity, we assumed them to be elements of $C^\infty(I;(0,\infty))$, although in many instances throughout this work less regularity would suffice. Note that we also allow for the case $n=2$, in which $\mathcal{M}$ is isometric to $\mathbb{R}\times I$ endowed with the line element $ds^2=-f(r)dt^2+h(r)dr^2$.

\begin{remark}
	We consider $I$ to be an open interval, say $I=(a,b)$, which might suggest that we are discarding scenarios of notable interest such as globally hyperbolic manifolds with timelike boundary, \textit{e.g.} the universal cover of $AdS_n$. In these cases, the counterpart of $I$ would include one or both endpoints in the domain of the coordinate $r$. Yet, if one is interested in the analysis of boundary conditions and their effects, it suffices to focus the attention on the interior of the underlying manifold. Therefore, our analyses can be straightforwardly applied to such cases as well.
\end{remark}

On $\mathcal{M}$, we consider a free, scalar field with mass $m_0\geq 0$, $\Psi:\mathcal{M}\rightarrow\mathbb{R}$ whose dynamics is ruled by the Klein-Gordon equation:
	\begin{equation}
		\label{eq: KG}
		P\Psi : = (\Box - m_0^2 - \xi \mathbf{R})\Psi=0,
	\end{equation}
where $\xi\in\mathbb{R}$, while $\mathbf{R}$ is the scalar curvature built out of the metric as per Equation \eqref{eq:metric schwarzschild-like}. In the case in hand, the D'Alembert wave operator, denoted by $\Box$, reads
$$\Box=-\frac{1}{f(r)}\partial^2_t + \frac{1}{h(r)}\partial^2_r + \frac{\partial_r(fh r^{n-4})}{fh r^{n-2}}\partial_r+ \frac{\Delta_{\Sigma_j^{n-2}}}{r^2},$$
where $\Delta_{\Sigma_j^{n-2}}$ is the Laplace operator on $\Sigma_j^{n-2}$. In the following, we construct the solutions of the Klein-Gordon equation. Although \Eref{eq: KG} can be recast in the form of \Eref{eq: KG evolution form}, we take a route alined with the procedure \ding{192}--\ding{197} described in the Introduction. Assuming that the regularity of $\Psi$ is such that we can work at the level of modes, we consider the ansatz
\begin{equation}
		\label{eq: ansatz KG in mode decomposition}
		\Psi_{\omega \eta_j}(t,r,\varphi_1,...,\varphi_{n-2})= e^{-i\omega t}R_{\omega \eta_j}(r)Y_j(\varphi_1,...,\varphi_{n-2}),
	\end{equation}
	where $Y_j(\varphi_1,...,\varphi_{n-2})$ are the eigenfunctions of $\Delta_{\Sigma_j^{n-2}}$ with corresponding eigenvalues denoted by $\eta_j$. Observe that $\Delta_{\Sigma_j^{n-2}}$ has a continuous spectrum if $j\in\{-1,0\}$, while a discrete one if $j=1$.

Equation \eqref{eq: ansatz KG in mode decomposition} in combination with \Eref{eq: KG} yields that $R_{\omega \eta_j}$ obeys to a second-order, ordinary differential equation, dubbed {\em radial equation}:
	\begin{align}
		& \mathbf{A}  R_{\omega \eta_j}(r) = 0 \label{eq: radial equation AR}\\
		& \mathbf{A}:= \frac{d^2}{dr^2} +   \left[ \frac{1}{2} \left(   \frac{f'(r)}{f(r)} - \frac{h'(r)}{h(r)} \right) + \frac{n-2}{r}  \right]  \frac{d}{dr}+ \left[\left(\frac{\omega ^2}{f(r)}+\frac{\eta_j }{r^2} -m_0^2 - \xi \mathbf{R}\right) h(r) \right]  .\nonumber
	\end{align}
	We can rewrite \Eref{eq: radial equation AR} as
	\begin{align}
		(&L-\lambda) R_{\omega \eta_j}(r) = 0, \label{eq: Sturm-Liouville radial eq} \\
		&L := \mathbf{A}  + \lambda = \frac{1}{\mu(r)} \left(- \frac{d}{d r} \left(p(r)\frac{d}{d r}\right) + q(r)\right),\nonumber
	\end{align}
	where
	\begin{equation}
		\label{eq: lambda eigenvalues for general SL schd-like}
		\lambda:=\begin{cases}
			\omega^2, \text{ if }f(r) \neq 1, \\
			\omega^2 - m_0^2 - \xi \mathbf{R}, \text{ if }f(r) = 1 \text{ and }  \mathbf{R} \text{ is constant}, \\
			\omega^2 - m_0^2, \text{ otherwise}.
		\end{cases}
	\end{equation}
	If $\lambda = \omega^2$, the functions $p$, $q$ and $\mu$ are given by
	\numparts
	\begin{align}\label{Eq: coefficients}
		& p(r) := r^{n-2 }\sqrt{\frac{f(r)}{h(r)}} =: \frac{r^{2n-4}}{\mu(r)},\\
		& q(r) := r^{n-4}(\eta_j - r^2 (m_0^2 + \xi \mathbf{R}) )\sqrt{f(r)h(r)}.\label{Eq: coefficients_2}
	\end{align}
	\endnumparts
	In the remaining cases the corresponding expressions can be derived directly from Equation \eqref{eq: radial equation AR}, but we omit listing them as they are a rather straightforward modification of Equation \eqref{Eq: coefficients} and \eqref{Eq: coefficients_2}. To conclude this section we observe that Equation \eqref{eq: radial equation AR} identifies a, possibly singular, Sturm-Liouville problem, following the standard nomenclature of ordinary differential equations, see {\it e.g.} \cite{zettl}.

\section{Self-adjoint extensions
	\label{sec: Self-adjoint extensions}
}

Envisioning the construction of ground states for the Klein-Gordon field $\Psi$, it is essential to obtain the advanced and retarded fundamental solutions associated to the operator $P$ as in Equation \eqref{eq: KG}. To this end we bear in mind a procedure that has been considered in several examples in the literature \cite{Dappiaggi2018xvw,Bussola2017wki,Dappiaggi2016fwc,Dappiaggi2022dwo,Campos2021wyr,deSouzaCampos2020bnj}, mainly when the underlying spacetime possesses a conformal, timelike boundary. The starting point is $\mathbf{A}$, as per Equation \eqref{eq: Sturm-Liouville radial eq}, which shall be read as an operator on the Hilbert space $\mathcal{L}^2(I, \mu(r)dr)$.

The operator $\mathbf{A}$, and consequently also $L$, is manifestly symmetric when taken with the dense domain $C^\infty_0(I)$. Herein, we scrutinize whether $L$ admits self-adjoint extensions and, if so, how many of them. This is a mathematical question that can be answered by combining tools of Sturm-Liouville theory with the theory of unbounded operators on Hilbert spaces, see {\it e.g.} \cite{Moretti}. Accordingly, in the following we recall the main results known in the literature that are of relevance to our investigation as well as necessary to make this work self-contained. All definitions and lemmas introduced here culminate in Theorem \ref{thm: theorem self-adjoint extensions generalized robin}, which constitutes the resolution to the question in hand.

First, let us pose the question more precisely. Consider the Sturm-Liouville problem, as in \Eref{eq: Sturm-Liouville radial eq},
\begin{align}
	\label{eq: sl problem self-adjoint xt}
	L R(r):=\frac{1}{\mu(r)} \left(- \frac{d}{d r} \left(p(r)\frac{d}{d r}\right) + q(r)\right) R(r)&= \lambda R(r),
\end{align}
where as of now we omit the subscripts $\omega,\eta_j$ from the radial function for decluttering. In addition, letting $\mathcal{L}^1_{\mini\text{loc}}$ refer to locally integrable functions, we assume that
\begin{equation*}
	1/p,q,\mu\in \mathcal{L}^1_{\mini\text{loc}}(I), \quad \mu>0,\quad \text{and}\quad r\in I:=(a,b) \text{, }-\infty \leq a < b \leq \infty,
\end{equation*}

As alluded to in the previous paragraphs, $L$ as per Equation \eqref{eq: sl problem self-adjoint xt} identifies either a {\bf minimal} or a {\bf maximal} operator respectively indicated by $L_{\mini{\text{min}}}$ and $L_{\mini{\text{max}}}$, with corresponding domains
	\begin{align*}
		&D_{\mini{\text{min}}}(L_{\mini{\text{min}}}) := \overline{C^\infty_0(I)},\\
		&D_{\mini{\text{max}}}(L_{\mini{\text{max}}}) := \{y:I\rightarrow\mathbb{C}\;|\; py'\in AC_{\mini{\text{loc}}}(I),\;\textrm{and}\; y,\mu^{-1}Ly\in\mathcal{L}^2(I, \mu(r)dr)\},
	\end{align*}
where the closure is taken with respect to the graph topology, and $AC_{\mini{\text{loc}}}(I)$ denotes the set of functions that are absolutely continuous on all compact intervals of $I$. Specifically, our quest is to find self-adjoint extensions $L_{\mini{\text{S.A.}}}$ of $L_{\mini{\text{min}}}$,
\begin{equation*}
	L_{\mini{\text{min}}}\subset L_{\mini{\text{S.A.}}}=L_{\mini{\text{S.A.}}}^*       \subset L_{\mini{\text{max}}},
\end{equation*}
whose domain shall be denoted by $D_{\mini{\text{S.A.}}}(L_{\mini{\text{S.A.}}})$. If their spectrum satisfies  $\sigma(L_{\mini{\text{S.A.}}})\subseteq [0,\infty)$, then we say $L_{\mini{\text{S.A.}}}$ is a \textbf{positive}, self-adjoint extension of $L_{\mini{\text{min}}}$. \\

In order to address the quest stated above, we introduce some additional tools tailoring the analysis of \cite{zettl} to the case of interest, {\it i.e.} Equation \eqref{eq: sl problem self-adjoint xt} together with the assumptions of the previous sections.

For $y,z\in AC_{\mini{\text{loc}}}(I)$ we denote the {\bf Lagrange sesquilinear form} and the {\bf Wronskian}, respectively, by
\begin{align}
	&[y,z] := p \left(y \overline{z}' -   y' \overline{z  } \right) = -\overline{[z,y]},\\
	&\{y,z\} := p \,\, \left(y z' -   y' z \right) = [y,\overline{z}].
\end{align}

\begin{definition}
	\label{def: weyls}
	Let $e$ be an endpoint of the interval $I=(a,b)$, {\it i.e.} $e\in\{a,b\}$. Define $I_c:=(a,c)$ if $e=a$, while $I_c:=(c,b)$ if $e=b$. We say the endpoint $e$ is
	\begin{enumerate}
		\item[i)] \textit{regular} if $1/p,q,\mu\in\mathcal{L}^1(I_c)$ for some (and hence any) $c\in I$;
		\item[ii)] \textit{singular} if it is not regular;
		\item[iii)] \textit{limit circle} if all solutions of \Eref{eq: sl problem self-adjoint xt} lie in $\mathcal{L}^2(I_c, \mu(r)dr),\, \forall c\in I$;
		\item[iv)] \textit{limit point} if it is not limit circle.
	\end{enumerate}
\end{definition}

\noindent The following definition is especially relevant for our analysis since it differentiates among the solutions of Equation \eqref{eq: sl problem self-adjoint xt} depending on their behaviour close to an endpoint.

\begin{definition}
	\label{def: principal and secondary}
	Let $y$ be a non-vanishing solution of \Eref{eq: sl problem self-adjoint xt} in $I_c,\forall c\in I$. Then we say $y$ is a
	\begin{enumerate}
		\item[i)] \textit{principal solution at $e$} if, for any other solution $z$ of Equation \eqref{eq: sl problem self-adjoint xt} that is linearly-independent from $y$,
		\begin{equation*}
			\frac{y(x)}{z(x)} \xrightarrow{x\rightarrow e}0.
		\end{equation*}
		\item[ii)] \textit{secondary (or non-principal) solution at $e$} if it is not a principal solution.
	\end{enumerate}
\end{definition}
Definitions \ref{def: weyls} and \ref{def: principal and secondary} relate by the fact that at a limit point only the principal solution belongs to $\mathcal{L}^2(I_c, \mu(r)dr),\, \forall c\in I$. Manifestly, the classification given by Definition \ref{def: principal and secondary} is of relevance only when at least one of the endpoints is a limit circle. Thus, for the remainder of this section, we assume that on $I=(a,b)$, $a$ is a limit circle while $b$ is a limit point. In addition, we denote by $u$ and $v$, respectively, the principal and secondary solutions at the limit circle endpoint and we set $[u,v](c)=\Lambda\in\mathbb{C}$ for all $c\in I$. We observe that although $a$ and $b$ might be singular endpoints, for any $y,z\in   D_{\mini{\text{max}}}(L_{\mini{\text{max}}}) $, since the following limits exist, it makes sense to define
\begin{equation}
	[y,z](a) := \lim\limits_{r\rightarrow a^+} [y,z]<\infty \quad \text{ and }\quad  [y,z](b) := \lim\limits_{r\rightarrow b^-} [y,z]<\infty.
\end{equation}

Next, we report three results concerning the interplay between the Lagrange sesquilinear form and Equation \eqref{eq: sl problem self-adjoint xt}. Their proofs, omitted here, are a direct adaptation to the case in hand respectively of Lemmas 10.4.1, 10.4.2 and 10.4.6 in \cite{zettl}.

\begin{lemma}
	\label{lem: Naimark Patching lemma}
	Let $L$ be as per Equation \eqref{eq: sl problem self-adjoint xt} and let $c,d\in I$. Given $\alpha,\beta,\gamma,\delta\in\mathbb{C}$, there exists $y\in D_{\mini{\text{max}}}(L_{\mini{\text{max}}})$ such that
	\begin{equation*}
		\begin{matrix*}[l]
			y(c) = \alpha, &  y(d) = \gamma, \\
			(py')(c) = \beta, & (py')(d) = \delta.
		\end{matrix*}
	\end{equation*}
\end{lemma}

\begin{lemma}(Lagrange bracket decomposition)
	\label{lem: Lagrange bracket decomposition}
 Let $L$ be as per Equation \eqref{eq: sl problem self-adjoint xt} and let $u,v$ be any two solutions lying in $D_{\mini{\text{max}}}(L_{\mini{\text{max}}})$ such that $[u,v](c)=\Lambda$ for all $c\in[a,b]$. It holds that, for all $y,z\in D_{\mini{\text{max}}}(L_{\mini{\text{max}}})$
	\begin{equation*}
		[y,z](c) = -\overline{\Lambda}^{-1}\left\{[y,v](c) \, [\overline{z},\overline{u} ](c)  -  [y,\overline{u}](c) \, [\overline{z},v](c) \right\}.
	\end{equation*}
\end{lemma}

\begin{lemma}
	\label{lem: that I need for thm}
Let $L$ be as per Equation \eqref{eq: sl problem self-adjoint xt}. Then for any $\lambda\in\mathbb{R}$ and $\alpha,\beta\in\mathbb{C}$, there exists $f\in D_{\mini{\text{max}}}(L_{\mini{\text{max}}})$ such that
	\begin{equation}
		\label{eq: lemma ioejfoijewoijf1}
		[f,u](a) = \alpha, \quad [f,v](a)=\beta.
	\end{equation}
	Moreover, if $a$ is a regular endpoint, then there exists $g\in D_{\mini{\text{max}}}(L_{\mini{\text{max}}})$ such that
	\begin{equation}
		\label{eq: lemma ioejfoijewoijf2}
		g(a) = \alpha, \quad pg'(a)=\beta.
	\end{equation}
\end{lemma}

\begin{remark}
	It is interesting to notice that if $\Lambda = 1$, \Eref{eq: lemma ioejfoijewoijf2} follows directly from \Eref{eq: lemma ioejfoijewoijf1} by setting $g=f$. This is not the case if $\Lambda\neq 1$ and this plays a significant part in the discussion of generalized versus regular boundary conditions in the next sections.
\end{remark}

\noindent The following result concerns properties of the self-adjoint extensions of $L_{\mini{\text{min}}}$, whereas their existence is a direct consequence of Von Neumann lemma \cite[Thm.5.43]{Moretti} since the differential operator $L$ in Equation \eqref{Eq: coefficients} has real coefficients. For its proof we refer to  \cite[Th. 10.4.1]{zettl}, and references therein.

\begin{lemma}
	\label{thm: sa domain d1}
	If $L_{\mini{\text{S.A.}}}$ is a self-adjoint extension of $L_{\mini{\text{min}}}$, then there exists $g\in D_{\mini{\text{S.A.}}}(L_{\mini{\text{S.A.}}})\subset D_{\mini{\text{max}}}(L_{\mini{\text{max}}})$ such that
	\begin{enumerate}
		\item[(1)] $g\notin D_{\mini{\text{min}}}(L{\mini{\text{min}}})$ and $[g,g](a)=0$;
		\item[(2)] $D_{\mini{\text{S.A.}}}(L{\mini{\text{S.A.}}}) = \{f\in D_{\mini{\text{max}}}(L_{\mini{\text{max}}}): [f,g](a)=0\}$.
	\end{enumerate}
	Conversely, for any $g\in D_{\mini{\text{max}}}(L_{\mini{\text{max}}})$ abiding to the conditions in item (1), there exists a self-adjoint extension of $L_{\mini{\text{min}}}$ whose domain  $D_{\mini{\text{S.A.}}}(L{\mini{\text{S.A.}}})$ is defined as per item (2).
\end{lemma}

Our quest reaches a finale with the following paramount result, which is specially tailored to befit singular Sturm-Liouville problems. In particular, it is instrumental to relating the existence of multiple self-adjoint extensions to the choice of specific boundary conditions. We include a detailed proof due to its relevance and in light of the fact that it is not exactly the well-known result as per \cite[Thm. 10.4.5]{zettl}, but rather a slight generalization of it. Namely, the principal and secondary solutions are not necessarily normalized to $[u,v]=1$.

\begin{theorem}
	\label{thm: theorem self-adjoint extensions generalized robin} Let $L$ be as in \Eref{eq: sl problem self-adjoint xt}. As per Definition \ref{def: principal and secondary}, let $u$ and $v$ be principal and secondary real-valued solutions at $r=a$ such that $[u,v]=\Lambda$. Then for any $(B_1,B_2)\in\mathbb{R}^2\setminus\{(0,0)\}$,
	\begin{equation}
		\label{eq: theorem self-adjoint extensions generalized robin}
		D_{\mini{\text{S.A.}}}(L_{\mini{\text{S.A.}}})= \{ y\in D_{\mini{\text{max}}}(L_{\mini{\text{max}}}): B_1 [y,u](a)+B_2 [y,v](a)  =  0 \}
	\end{equation}
	identifies the domain of a self-adjoint extension of $L_{\mini{\text{min}}}$. Moreover, all self-adjoint extensions of $L_{\mini{\text{min}}}$ are of this form.
\end{theorem}

\begin{proof} 
We divide the analysis in two separate parts: proof of the first statement, and proof of the ``moreover'' statement.
	\vskip .2cm
\noindent{\em Part $1.$}
		\begin{itemize}
		\item[] Let $(B_1,B_2)\in\mathbb{R}^2\setminus\{(0,0)\}$. To prove that \Eref{eq: theorem self-adjoint extensions generalized robin} identifies the
		the domain of a self-adjoint extension of $L_{\mini{\text{min}}}$, we shall use Lemma \ref{thm: sa domain d1}. In other words we set $g = \frac{B_1 u + B_2 v}{\Lambda }$ and, using Lemma \ref{lem: Naimark Patching lemma}, we can make sure that $g,u,v\in D_{\mini{\text{max}}}(L_{\mini{\text{max}}})\setminus D_{\mini{\text{min}}}(L_{\mini{\text{min}}})$. It descends
		\begin{equation*}
			[g,g](a) = \frac{1}{\Lambda} (B_1 [g,u](a) + B_2 [g,v](a)) =  0, \quad [g,u](a) = - B_2 \quad [g,v](a) =   B_1.
		\end{equation*}
		Observe that item (2) of Lemma \ref{thm: sa domain d1} is automatically fulfilled by \Eref{eq: theorem self-adjoint extensions generalized robin}.
		\end{itemize}
	\vskip .2cm
{\em Part $2.$}
\begin{itemize}
		\item[]  Consider now $L_{\mini{\text{S.A.}}}$ a self-adjoint extension of $L_{\mini{\text{min}}}$. By Lemma \ref{lem: Lagrange bracket decomposition}, for $y,g\in D_{\mini{\text{max}}}(L_{\mini{\text{max}}})$, it holds
		\begin{align*}
			[y,g](a) &=  -\frac{1}{\Lambda}\left\{[y,v](a) \, [\overline{g},\overline{u} ](a)  -  [y,\overline{u}](a) \, [\overline{g},v](a) \right\}\\
			& =  -\frac{1}{\Lambda} \left\{[y,v](a) \,C_1 -  [y,\overline{u}](a) \, C_2 \right\},
		\end{align*}
		where $C_1:=[\overline{g},\overline{u} ](a)$ and $C_2:=[\overline{g},v](a)$. On account of Lemma \ref{thm: sa domain d1} there exists $g\notin D_{\mini{\text{min}}}(L_{\mini{\text{min}}})$ such that $[g,g](a)=0$. By setting $y=g$ in the equation displayed above and assuming both $u$ and $v$ to be real-valued, it descends
		\begin{align*}
			[g,g](a) =  -\frac{1}{\Lambda} \left\{[g,v](a) \,C_1 -  [g,u](a) \, C_2 \right\}
			= -\frac{1}{\Lambda} \left\{ \overline{C_2}  \,C_1 -  \overline{C_1} \, C_2 \right\}.
		\end{align*}
		Thus
		\begin{equation*}
			[g,g](a) = 0 \iff \overline{C_1} \, C_2 \in\mathbb{R}.
		\end{equation*}
	In addition, still Lemma \ref{thm: sa domain d1} guarantees that $f\in D_{\mini{\text{S.A.}}}(L_{\mini{\text{S.A.}}})$ if and only if $[f,g](a)=0$.
	This reduces to \Eref{eq: theorem self-adjoint extensions generalized robin} setting therein $B_1 = \overline{C_1} C_1$ and $B_2 =  \overline{C_1} C_2$.
\end{itemize}
\end{proof}
\begin{remark} Note that the proof of Theorem \ref{thm: theorem self-adjoint extensions generalized robin} assuming $[u,v]=\Lambda$ is analogous to that on \cite[Thm. 10.4.5]{zettl} that assumes $[u,v]=1$. It is worth mentioning that this normalization does not select a secondary solution. It is easy to see that this is the case if we take into account that $[u,u]=0$ for real-valued $u$. In turn, the reality of both $u$ and $v$ is an essential aspect of the validity of the proof. In addition, a consequence of this restriction is that \Eref{eq: theorem self-adjoint extensions generalized robin} can be equivalently written in terms of the Wronskians instead of the Langrage sesquilinear form, {\it i.e.}
 \begin{equation*}
  B_1\{y,u\}(a)+B_2 \{y,v\}(a) = 0.
 \end{equation*}
\end{remark}

To summarize, Theorem \ref{thm: theorem self-adjoint extensions generalized robin} characterizes all self-adjoint realizations of the Sturm-Liouville problem under consideration. Namely, $L_{\mini{\text{S.A.}}} \,y   = L_{\mini{\text{S.A.}}}^* y$ for $y\in \mathcal{L}^2(I_c, \mu(r)dr)$ such that
$$	 B_1 \{y,u\}(a) + B_2 \{y,v\}(a) = 0 \text{ for }  (B_1, B_2)\in\mathbb{R}^2\setminus\{(0,0)\}.$$

\section{Generalized $(\gamma,v)$-Robin boundary conditions
	\label{sec: generalized gamma v robin}
}

In this section we take a closer look at the boundary condition stated in \Eref{eq: theorem self-adjoint extensions generalized robin} and we reiterate two important facts:
\begin{itemize}
	\item although \Eref{eq: theorem self-adjoint extensions generalized robin} depends on the choice of the pair $(B_1,B_2)\neq(0,0)$, it is always possible to rescale $y\in D_{\mini{\text{max}}}(L_{\mini{\text{max}}})$ so to fix one of the parameters to $1$, {\it i.e.} we can consider only pairs of the form $(1,\frac{B_2}{B_1})$;
	\item the characterization of $D_{\mini{\text{S.A.}}}(L_{\mini{\text{S.A.}}})$ \textit{also depends on the chosen secondary solution}.
\end{itemize}

At the mere level of the Sturm-Liouville problem under consideration, the freedom in the choice of secondary solution is inconsequential if one is interested in characterizing all self-adjoint extensions of the corresponding operator. Notwithstanding, it plays a distinguished, physically relevant r\^{o}le when we turn back to analyzing the dynamics of the Klein-Gordon field ruled by Equation \eqref{eq: KG}. The following definitions aim at highlighting this freedom in the overall process and the difference that occurs when considering a singular rather than a regular Sturm-Liouville problem.

\begin{definition}(Regular $\gamma$-Robin boundary condition)	\label{def: regular Robin}
	Let $L$ be as per \Eref{eq: sl problem self-adjoint xt}. Given any self-adjoint realization $L_{\mini{\text{S.A.}}}$, we say that $y\in D_{\mini{\text{S.A.}}}(L_{\mini{\text{S.A.}}})$ satisfies a {\bf regular $\gamma$-Robin boundary condition} at $a$ if
	\begin{equation*}
		\lim\limits_{r\rightarrow a} \left\{ \cos(\gamma) y + \sin(\gamma) y^\prime\right\} = 0 \quad \text{ for } \quad   \gamma \in [0,\pi),
	\end{equation*}
where the prime indicates the derivative along the $r$-direction. In particular, we say that $y$ abides to a
	\begin{enumerate}
		\item {\em regular Dirichlet boundary condition at $a$} if it satisfies a regular $0$-Robin boundary condition:
		\begin{equation*}
			\lim\limits_{r\rightarrow a}  y = 0  \quad \text{ and } \quad     \lim\limits_{r\rightarrow a} y' = c \in \mathbb{R}.
		\end{equation*}
		\item {\em regular Neumann boundary condition at $a$} if it satisfies a regular $\frac{\pi}{2}$-Robin boundary condition:
		\begin{equation*}
			\lim\limits_{r\rightarrow a} y' = 0   \quad \text{ and } \quad       \lim\limits_{r\rightarrow a} y  = c \in \mathbb{R} .
		\end{equation*}
	\end{enumerate}
\end{definition}

\begin{definition}(Generalized $(\gamma,v)$-Robin boundary condition)
	\label{def: generalized Robin}
		Let $L$ be as per \Eref{eq: sl problem self-adjoint xt} and let $u$ be the principal solution at $a$ and $v$ any secondary solution at $a$, real-valued and such that $\{u,v\}=\Lambda$. Given any self-adjoint realization $L_{\mini{\text{S.A.}}}$, we say that $y\in D_{\mini{\text{S.A.}}}(L_{\mini{\text{S.A.}}})$ satisfies a \textbf{generalized $(\gamma,v)$-Robin boundary condition at $a$ } if
		\begin{equation}
			\label{eq: definition generalized Robin}
			\lim\limits_{r\rightarrow a} \left\{  \cos(\gamma)\{y,u\} +  \sin(\gamma)\{y,v\} \right\}=0   \quad \text{ for } \quad   \gamma \in [0,\pi).
		\end{equation}
	In particular, we say that $y$ abides to a
	\begin{enumerate}
		\item \textit{generalized Dirichlet boundary condition at $a$} if it satisfies a generalized $(0,v)$-Robin boundary condition:
		\begin{equation}
			\label{eq: definition generalized Dirichlet}
			\lim\limits_{r\rightarrow a} \{y,u\} =0.
		\end{equation}
		\item \textit{generalized $v$-Neumann boundary condition at $a$} if it satisfies a generalized $(\frac{\pi}{2},v)$-Robin boundary condition:
		\begin{equation}
			\label{eq: definition generalized Neumann}
			\lim\limits_{r\rightarrow a} \{y,v\} =0.
		\end{equation}
	\end{enumerate}
\end{definition}

It is worth stressing that Definition \ref{def: regular Robin} is applicable only to {\em regular} Sturm-Liouville problems since it implicitly requires differentiability of the solution at the endpoint $a$. In addition, consistently with what one could a priori expect, the ``Dirichlet boundary condition'' is actually independent of the choice of the secondary solution.  In the following, we elucidate more in detail the connection between the two definitions above.

\subsection{Reduction to the regular case}

Consider the setting of Theorem \ref{thm: theorem self-adjoint extensions generalized robin}, and assume that $r=a$ is a regular endpoint as per Definition \ref{def: weyls}. For real-valued $u$ and $v$ such that $\{u,v\}=\Lambda$, a solution $y = \cos(\gamma) u + \sin(\gamma) v \in\mathcal{L}^2(I, \mu(r)dr)$ satisfies a generalized $(\gamma,v)$-Robin boundary condition, as per \Eref{eq: definition generalized Robin}. Hence
\begin{equation}
	\label{eq:owijefoiwjoief}
	[y,u](a) = -\Lambda \sin(\gamma), \quad [y,v](a)= \Lambda\cos(\gamma).
\end{equation}
On account of Lemma \ref{lem: that I need for thm}, since $a$ is a regular endpoint, if we choose $u$ and $v$ such that
\begin{equation*}
	u(a) = 0 \quad \text{ and }\quad   (pv')(a) = 0,
\end{equation*}
then \Eref{eq:owijefoiwjoief} implies that
\begin{equation*}
	y(a) = -\frac{ \Lambda \sin(\gamma)}{pu'(a)} \quad \text{ and }\quad y'(a)= -\frac{\Lambda\cos(\gamma)}{p v(a)}.
\end{equation*}
Hence
\begin{equation*}
	\cos(\gamma)y(a) + \sin(\gamma)y'(a) = - \frac{\Lambda \sin(\gamma)\cos(\gamma)}{p} \left(\frac{1}{u'(a)} + \frac{1}{v(a)} \right)=0\iff v(a) =-u'(a).
\end{equation*}
That is, at a regular endpoint, for a given $u$ there is a choice of secondary solution $v$ for which the generalized $(\gamma,v)$-Robin boundary condition, as per Definition \ref{def: generalized Robin}, yields a regular $\gamma$-Robin boundary condition, as per Definition \ref{def: regular Robin}. Conversely, if we do not choose it in such a way and if $\gamma\neq 0$, then a generalized boundary condition does not necessarily reduce to a regular one.

\subsection{Generalized $(\gamma,v)$-Robin boundary conditions and the Klein-Gordon equation}\label{sec: Sensible dynamics}

In view of the foregoing discussion, it is natural to wonder what is the consequence of choosing a specific generalized $(\gamma,v)$-Robin boundary condition at the level of the fully covariant Klein-Gordon equation. This question becomes especially relevant when we are working on a globally hyperbolic spacetime with timelike boundary \cite{Ake2018dzz}, such as the Poincar\'e patch of an $n$-dimensional anti-de Sitter spacetime (PAdS$_n$). In this case, it is known that the dynamics is completely specified when, and only when, initial data are supplemented with a boundary condition assigned at conformal infinity. Then, one would slavishly follow the analysis outlined in the previous sections.

For definiteness, let us assume we are working under conditions for which Definition \ref{def: generalized Robin} is meaningful. Explicitly, take $\Psi$ to be a solution of the Klein-Gordon equation \eqref{eq: KG} written as the mode expansion given in \Eref{eq: ansatz KG in mode decomposition}. In addition, let $u$ and $v$ be real-valued principal and secondary solutions at an endpoint $a$ for the radial equation, such that the radial mode satisfies a generalized $(\gamma,v)$-Robin boundary condition. One can infer that the latter translates to a boundary condition on $\Psi$:
\begin{equation}
	\label{eq:bc psi generalized robin}
	\cos(\gamma) \{\Psi,\widetilde{u}\}(a) + \sin(\gamma)\{\Psi, \widetilde{v}\}(a) = 0,
\end{equation}
where
\begin{align*}
	&\widetilde{u} =  \int_{\sigma(\Delta_j^{n-2})}d\mu(\Sigma_j^{n-2}) \int_{\sigma(\mathbf{A})}d\lambda \,e^{-i\sqrt{\lambda} t}u(r)Y_j(\varphi_1,...,\varphi_{n-2}), \\
	&\widetilde{v} =  \int_{\sigma(\Delta_j^{n-2})}d\mu(\Sigma_j^{n-2}) \int_{\sigma(\mathbf{A})}d\lambda \,e^{-i\sqrt{\lambda} t}v(r)Y_j(\varphi_1,...,\varphi_{n-2}),
\end{align*}
where $\int_{\sigma(\Delta_j^{n-2})}d\mu(\Sigma_j^{n-2})$ is the integral over the spectrum of the Laplace operator on $\Sigma_j^{n-2}$. Details regarding the latter are left to the reader since they play no r\^{o}le in our discussion, yet we observe that if $j=1$ this integral reduces to a sum of (hyper-)spherical harmonics, whereas if $j\in\{0,1\}$ it is nothing but an ordinary Lesbegue integral. Similarly, $\int_{\sigma(\mathbf{A})}d\lambda$ is formally the integral over the spectrum
of the self-adjoint extension $\mathbf{A}$ with respect to the associated spectral measure. For all practical purposes in many instances $\sigma(\mathbf{A})=(0,\infty)$ and the integral reduces to a standard Lesbegue integration on the half real line.

With the discussion of Section \ref{sec: generalized gamma v robin} in mind, we see that \Eref{eq:bc psi generalized robin} reduces to a regular Robin boundary condition only under special conditions. In addition, we highlight that generalized $(\gamma,v)$-Robin boundary conditions on $R(r)$ translate at the level of the Klein-Gordon equation, to a wide variety of boundary conditions, including time-dependent ones. Although at this stage, this statement might be elusive and hidden in the meanders of \Eref{eq:bc psi generalized robin}, it is manifest in the concrete example thoroughly discussed in Section \ref{sec: Time-dependence of the boundary conditions HO}.

\section{The quantum dynamics
	\label{sec: The quantum dynamics}
}
The analysis of the classical solutions to the Klein-Gordon equation is just the starting point to obtain a full-fledged, covariant quantization framework. In this paper we shall not give all the details of the latter, see \cite{Dappiaggi2018xvw,Bussola2017wki,Dappiaggi2016fwc,Dappiaggi2022dwo,Campos2021wyr,deSouzaCampos2020bnj}, rather we focus on the construction of ground states admitting generalized $(\gamma,v)$-Robin boundary conditions.

\begin{definition}\label{Def: two-point function}
	Let $\mathcal{M}$ and $P$ be as defined in Section \ref{sec: The Klein-Gordon equation}. A two-point function of a quantum state is a bidistribution $\psi_2\in\mathcal{D}'(\mathcal{M}\times\mathcal{M})$ such that
	\begin{enumerate}
		\item it solves the Klein-Gordon equation in both entries: $$(P\otimes\mathbbm{1}) \psi_2 = (\mathbbm{1}\otimes P) \psi_2 =0;$$
		\item it satisfies the {\em canonical commutation relations}:
		$$E(f,f'):=\psi_2(f,f') -  \psi_2(f',f),\quad\forall f,f^\prime\in C^\infty_0(M),$$
		where $E$ is the advanced minus retarded fundamental solution associated to $P$; 
		\item it is positive:
		$$\frac{1}{4}\left|E(f,f^\prime)\right|^2\leq\psi_2(f,f)\psi_2(f^\prime,f^\prime)\quad\forall f,f^\prime\in C^\infty_0(M).$$
		\end{enumerate}
\end{definition}

 In turn, the bi-distribution $E\in\mathcal{D}^\prime(\mathcal{M}\times\mathcal{M})$ is a solution of the initial value problem:
\begin{subequations}
\begin{align}\label{eq: causal propagator}
&(P\otimes\mathbbm{1}) E = (\mathbbm{1}\otimes P) E =0,\\
&E|_{\mini{\Sigma_t\times \Sigma_t}} =0    \text{   and   }  \nabla_\xi E|_{\mini{\Sigma_t\times \Sigma_t}} =\delta_{\Sigma_t},
\end{align}
\end{subequations}
where $\Sigma_t$ is any constant-time hypersurface, while $\delta_\Sigma$ is the Dirac delta thereon. It is important to stress that $E$ is a priori not unique, depending both on the underlying geometry and on the parameters $\xi$ and $m_0$ of the Klein-Gordon equation, see \Eref{eq: KG}. The details for its construction using tools of Sturm-Liouville and spectral theories can be found in \cite{DeSouzaCampos2022wsp} and references therein.

As detailed in \cite[Ch.2]{DeSouzaCampos2022wsp} and hinted at in the Introduction, among the plethora of two-point functions on a spacetime admitting Schwarzschild-like coordinates, as per \Eref{eq:metric schwarzschild-like}, one can always distinguish the ones that characterize {\em ground states}. They are of the form
\begin{align}
	\label{eq: def 2 point ground state schd coord}
	\psi_2(x,x^\prime)= &  \lim_{\varepsilon\rightarrow 0^+} \int_{\sigma(\Delta_{j})}d\eta_j\int_{\mathbb{R}}d\omega \Theta(\omega) e^{-i\omega (t - t'- i\varepsilon)} \widetilde{\psi}_{2}(r,r')Y_{j}(\underline{\theta})\overline{Y_{j}(\underline{\theta}')},
\end{align}
where $\Theta$ denotes the Heaviside step function and, for compactness, we have introduced the notation $\underline{\theta}=(\varphi_1,...,\varphi_{n-2})$. Using \Eref{eq: causal propagator} in combination with the canonical commutation relations in Definition \ref{Def: two-point function} and with the completeness of the eigenefunctions of the Laplace operator $\Delta_{\Sigma_j^{n-2}}$, it turns out that the unknown $ \widetilde{\psi}_2(r,r^\prime)$ can be obtained by the spectral resolution of the Green function $\mathcal{G}(r,r') $ associated to \Eref{eq: radial equation AR}, see \cite[Ch.7]{greenBook}. Namely, promoting $\lambda$ to a complex variable, one make use of the chain of identities
\begin{equation}
	\label{eq:spectral resolution radial green function}
	\int_{\mathbb{R}}\omega d \omega  \widetilde{\psi}_2(r,r^\prime)   = \frac{\delta(r-r')}{\mu(r)} = -\frac{1}{2\pi i}\oint_{\mathcal{C}^\infty} d\lambda \mathcal{G}(r,r'),
\end{equation}
where $\mathcal{C}^\infty$ is an infinitely large circle in the $\lambda$-plane with a counter-clockwise orientation.

\begin{remark}
Each generalized $(\gamma,v)$-Robin boundary condition on the radial mode $R(r)$ yields a different Green function $\mathcal{G}(r,r') $ for the radial equation and, consequently, a different two-point function---a different ground state.
\end{remark}

In the next section we give a neat and tangible example that unveils how the choice of different secondary solutions, even with the same value of $\gamma$, yields well-defined but inequivalent ground states. It corroborates our statement that, in a system where states can be constructed following a mode decomposition, the choice of a secondary solution for the radial equation remains free even after imposing all the physical constraints necessary to guarantee a sensible framework.

\section{An illustrative example: the wave equation on $\mathbb{R}\times\mathbb{R}_+$}
	\label{sec: example harmonic oscillator}

In this section, we outline a simple yet most illustrative example aimed at highlighting the relevance of the generalized $(\gamma,v)$-Robin boundary conditions: a massless, real, scalar field on the $2$-dimensional half-Minkowski spacetime $\mathbb{R}\times\mathbb{R}_+$. Although the endpoints to be considered in the corresponding Sturm-Liouville problem are regular, since the underlying manifold is globally hyperbolic with a timelike boundary, it is conceivable to impose thereon generalized $(\gamma,v)$-Robin boundary conditions. 

In standard Cartesian coordinates $(t,x)$, for $t\in\mathbb{R}$ and $x\in(0,\infty)$, the mode decomposition ansatz as per \Eref{eq: ansatz KG in mode decomposition} reads
\begin{equation*}
	\Psi_\omega(t,x) = e^{-i\omega t} \psi_\omega(x),
\end{equation*}
where $\psi_\omega(x)$ is nothing but an harmonic oscillator on the half-line:
\begin{equation*}
	\frac{d^2 \psi(x)}{dx^2} = - \omega^2 \psi(x), \quad \forall x\in (0,\infty).
\end{equation*}
Consider the basis of solutions given by 
\begin{equation*}
	y_1(x) := e^{+i\omega x} \quad \text{ and }\quad  y_2(x) := e^{-i\omega x}.
\end{equation*}
As functions of $x$ and considering $\omega$ as a possibly complex parameter rather than a Fourier variable conjugated to time $t$, it follows that $y_1(x)$ and $y_2(x)$ are not square-integrable at $x=\infty$ unless $\Imag(\omega)>0$ and $\Imag(\omega)<0$, respectively. Therefore, as per Definition \ref{def: weyls}, this endpoint is a limit point 
and the most general square-integrable solution therein can be written concisely as
\begin{equation}
	\label{eq: square-integrable solution at infinity harmonic oscillator}
	\psi_\infty(x) =  y_1(x) \Theta(\Imag(\omega)) +  y_2(x) \Theta(-\Imag(\omega)).
\end{equation}
Both $y_1(x) $ and $y_2(x)$ are square-integrable in a neighborhood of $x=0$. Still according to Definition \ref{def: weyls}, $x=0$ is a limit circle.

Using the nomenclature of Definition \ref{def: principal and secondary}, the principal solution reads
\begin{equation*}
	u(x) = \frac{y_1(x) - y_2(x)}{2i}= \sin(\omega x) .
\end{equation*}
As a secondary solution at $x=0$, we shall consider three possible choices:
\begin{align*}
	v_1 &=\omega \frac{y_1(x) + y_2(x)}{2}= \omega \cos(\omega x), \\
	v_2 &=\frac{y_1(x) + y_2(x)}{2}= \cos(\omega x),\\
	v_3 &=\frac{y_1(x) + y_2(x)}{2} + \frac{y_1(x) - y_2(x)}{2i}= \cos(\omega x) + \sin(\omega x).
\end{align*}
For convenience, define $s := \text{sign}(\Imag(\omega))$, $\Imag(\omega)\neq 0$. We can write $\psi_\infty$ in terms of the principal and secondary solutions as
\begin{equation*}
	\psi_\infty(x) = a_\kappa u + b_\kappa v_\kappa,
\end{equation*}
where $a_1 = a_2 = a_3 + 1 = i\,s$ and $\omega b_1 = b_2 = b_3 = 1$.
Note that $  \psi_\infty(x)$ itself is independent of $\kappa$.
In addition, the solution
\begin{equation}
	\label{eq: square-integrable solution at zero harmonic oscillator}
	\psi_\kappa (x)= \cos(\gamma) u(x) + \sin(\gamma) v_\kappa(x), \quad \text{for }\kappa\in\{1,2,3\}, \, \gamma\in\mathbb{R},
\end{equation}
satisfies a generalized $(\gamma,v_\kappa)$-Robin boundary condition at $x=0$, {\it i.e.}
\begin{equation*}
	\lim\limits_{x\rightarrow 0} \left\{ \cos(\gamma) \{\psi_\kappa, u\} + \sin(\gamma) \{\psi_\kappa,v_\kappa]\}\right\} =0.
\end{equation*}

\subsection{Generalized versus regular Robin boundary conditions}
Since the limits of $y_1$, $y_2$ and of their derivatives exist as $x\rightarrow 0$, we can cast the generalized $(\gamma,v_\kappa)$-Robin boundary condition above as a regular $\gamma$-Robin boundary condition:
\begin{equation*}
	\lim\limits_{x\rightarrow 0}\left(\psi_\kappa - \frac{1}{\beta_\kappa}\psi_\kappa'\right) = 0, \quad \text{ for }\quad \beta_\kappa = \frac{\cos(\gamma) u'(0) + \sin(\gamma) v_\kappa'(0)}{\cos(\gamma) u(0) + \sin(\gamma) v_\kappa(0)}.
\end{equation*}
We find that
\begin{align*}
	\beta_1 &= \cot(\gamma), \\
	\beta_2 &= \omega \cot(\gamma),\\
	\beta_3 &= \omega (\cot(\gamma) +1),
\end{align*}
which highlights the difference between a generalized and a regular Robin boundary condition. Markedly, in the regular scenario we can choose a secondary solution based on the property of obtaining a frequency-independent parameter $\beta_\kappa$; in this case, $\beta_1$. Yet, frequency-dependent boundary conditions are also of physical relevance and hence there is no a priori reason to discard them.

\subsection{Time-dependence of the boundary conditions
	\label{sec: Time-dependence of the boundary conditions HO}}
Analogously to the discussion in Section \ref{sec: Sensible dynamics}, given a radial solution of the wave function that satisfies a generalized $(\gamma,v_\kappa)$-Robin boundary condition, it is legitimate to wonder which is the corresponding boundary condition satisfied by the solution of the wave equation on $\mathbb{R}\times\mathbb{R}_+$. For $\psi_\kappa$ given by \Eref{eq: square-integrable solution at zero harmonic oscillator}, let us consider a general solution:
\begin{equation*}
	\Psi_{\kappa}(t,x) = \int\limits_{\mathbb{R}}d\omega e^{-i\omega t} \psi_{\kappa}(x).
\end{equation*}
It can be written as
\begin{equation}
	\label{eq: full solution harmonic oscillator}
	\Psi_{\kappa}(t,x) = \cos(\gamma) \widehat{u} +\sin(\gamma) \widehat{v_\kappa},
\end{equation}
where, with a slight abuse of notation, we have denoted the Fourier transform with a hat. We find
\begin{align*}
	&\widehat{u} = i\pi (\delta(t+x) - \delta(t-x)),\\
	&\widehat{v_1} = i\pi (\delta'(t+x) + \delta'(t-x))\\
	&\widehat{v_2} = \pi (\delta(t+x) + \delta(t-x))\\
	&\widehat{v_3} = \pi ((1+i)\delta(t+x) + (1-i)\delta(t-x).
\end{align*}
We can now read which boundary condition is satisfied by $\Psi_{\kappa}(t,x)$ at $x=0$ for each $\kappa$. Formally, working at the level of distributions, we look for operators such that
\begin{align*}
	\lim\limits_{x\rightarrow 0}  (c_1 \Psi_{1}(t,x) + \partial_x \Psi_{1}(t,x)) &= 2i\pi \delta'(t)(c_1 \sin(\gamma)+\cos(\gamma)) =0;\\
	\lim\limits_{x\rightarrow 0} (c_2 \Psi_{2}(t,x) +  \partial_x \Psi_{2}(t,x)) & = 2i\pi (c_2 \delta(t) \sin(\gamma)+ i \delta'(t) \cos(\gamma)) =0;\\
	\lim\limits_{x\rightarrow 0} ( c_3  \Psi_{3}(t,x) +\partial_x \Psi_{3}(t,x)) &= 2\pi (c_3 \delta(t) \sin(\gamma) +  i \delta'(t) (\cos(\gamma) + \sin(\gamma)) =0.
\end{align*}
The above are satisfied if we take $c_\kappa$ and $\zeta_\kappa$ as
\begin{align*}
	&c_1=1,\\
	&c_2 =  -i\cot(\gamma) \partial_t =:\zeta_2 \partial_t,\\
	&c_3 =-i \cot(\gamma) (1+\tan(\gamma))\partial_t =: \zeta_3 \partial_t .
\end{align*}
Accordingly, the solution $\Psi_\kappa(t,x)$ satisfies, for $\kappa \in\{2,3\}$:
\begin{equation*}
	\lim\limits_{x\rightarrow 0}  (\zeta_\kappa \partial_{t} \Psi_{\kappa}(t,x) + \partial_x \Psi_{\kappa}(t,x)) =0.
\end{equation*}
That is, the solutions $\Psi_2(t,x)$ and $\Psi_3(t,x)$ do not satisfy a regular, time-independent, $\gamma$-Robin boundary condition at the boundary as it does $\Psi_1(t,x)$.


\subsection{The Green functions}
To construct the ground state for a Klein-Gordon field admitting generalized $(\gamma,v_\kappa)$-Robin boundary conditions on the two-dimensional half-Minkowski spacetime, we follow the rationale outlined in Section \ref{sec: The quantum dynamics}. On account of \Eref{eq:spectral resolution radial green function}, the only unknown is $\widetilde{\psi}_2(x,x^\prime)$, which in turn can be constructed from the Green function $\mathcal{G}(x,x^\prime)$ of the radial equation. Note that, in the case in hand, the r\^{o}le of the radial coordinate $r$ is played by the cartesian coordinate $x$.

Let the Wronskian between the principal and the secondary solutions be given by $ \{u,v_\kappa\} =: \Lambda_\kappa$. 
Then the one between the two general square-integrable solutions given by \Eref{eq: square-integrable solution at infinity harmonic oscillator} and \Eref{eq: square-integrable solution at zero harmonic oscillator} reads
\begin{align*}
	\{\psi_\kappa(x),\psi_\infty(x)\} 
	&= \Lambda_\kappa  (b_\kappa \cos(\gamma) - a_\kappa \sin(\gamma)) .
\end{align*}
Therefore, the Green functions $\mathcal{G}_\kappa(x,x')$ for the different choices of secondary solutions, which satisfy
\begin{equation*}
	\left(\frac{d^2}{dx^2} + \omega^2 \right) \mathcal{G}_\kappa(x,x') = \left(\frac{d^2}{dx'^2} + \omega^2 \right) \mathcal{G}_\kappa(x,x') = \delta(x-x'),
\end{equation*}
are given by
\begin{equation}
	\label{eq: definition green function harmonic oscillator general k}
	\mathcal{G}_\kappa(x,x') := \frac{\psi_\kappa(x_<) \psi_\infty(x_>)}{\{\psi_\kappa,\psi_\infty\}}, \quad \text{for } \omega^2\notin [0,\infty),
\end{equation}
where $(x_<, x_>) = (x, x')$ if $x<x'$ and  $(x_<, x_>) = (x', x)$, otherwise. Observe that the dependence on $\omega$ is implicit in the solutions $\psi_\kappa$ and $\psi_\infty$. Moreover, for $\gamma\in\mathbb{R}$, it holds true that
\begin{align*}
	(b_\kappa \cos(\gamma) -   a_\kappa \sin(\gamma)) = 0 \iff \cot(\gamma) = \frac{a_\kappa}{b_\kappa} \iff k=1 \text{ and } \gamma\in\left(\frac{\pi}{2},\pi\right).
\end{align*}
For notational convenience, but with no loss of generality, let us consider $x<x'$ and $\gamma\in [0,\frac{\pi}{2}]$ in the remainder of this section. Once more we fix $s := \text{sign}(\Imag(\omega))$, $\Imag(\omega)\neq 0$. Explicitly, the Green functions can be written as
\numparts
\label{eq: green functions exponential form harmonic oscillator explicitly all three k}
\begin{align*}
	\mathcal{G}_1(x,x') &=  \frac{-i s}{2\omega} \left\{ e^{-i s \omega (x-x')} +  e^{i s \omega (x+x')}\left[1 - \frac{2}{1-is\omega \tan(\gamma) } \right] \right\},\\
	\mathcal{G}_1(x,x') &=  \frac{-i s}{2\omega}  \left\{ e^{-i s \omega (x-x')} + e^{i s \omega (x+x')}\left[1 - \frac{2}{1-is \tan(\gamma) } \right]  \right\},\\
	\mathcal{G}_3(x,x') &=  \frac{-i s}{2\omega} \left\{ e^{-i s \omega (x-x')} +  e^{i s \omega (x+x')}\left[1 - \frac{2}{1+is(1+\cot(\gamma)) } \right] \right\}.
\end{align*}
\endnumparts

\subsection{The resolution of the identity}

For all three values of $\kappa$, the suitable contour to be considered for the integration of $\mathcal{G}_\kappa$ is the ``pac-man'' in the $\omega^2$-complex plane, which is tantamount to integrating $2 \omega \mathcal{G}_\kappa$ in one semi-disk in the upper or lower $\omega$-complex plane, as illustrated in Figure \ref{fig:contour}. Note that although $\mathcal{G}_1$ and $\mathcal{G}_3$ diverge in the limit $\omega\rightarrow0$, $\omega \mathcal{G}_\kappa$ has no poles in the $\omega$-complex plane. In addition, since Jordan's lemma holds true, it follows that
\begin{align}
	\label{eq: contour integral equal real-line integral}
	-\frac{1}{2\pi i}\oint\limits_{\text{"pac-man"}}& \mathcal{G}_\kappa(x,x')d\omega^2 =\lim\limits_{R\rightarrow\infty} \int\limits_{-R}^{R}\omega\,d\omega \lim\limits_{\varepsilon\rightarrow 0^+} [\mathcal{G}_\kappa(x,x')|_{\omega+i\epsilon}-\mathcal{G}_\kappa(x,x')|_{\omega-i\epsilon}] \\
\end{align}
\begin{figure}[H]
	\centering
	\includegraphics[width=.8\textwidth]{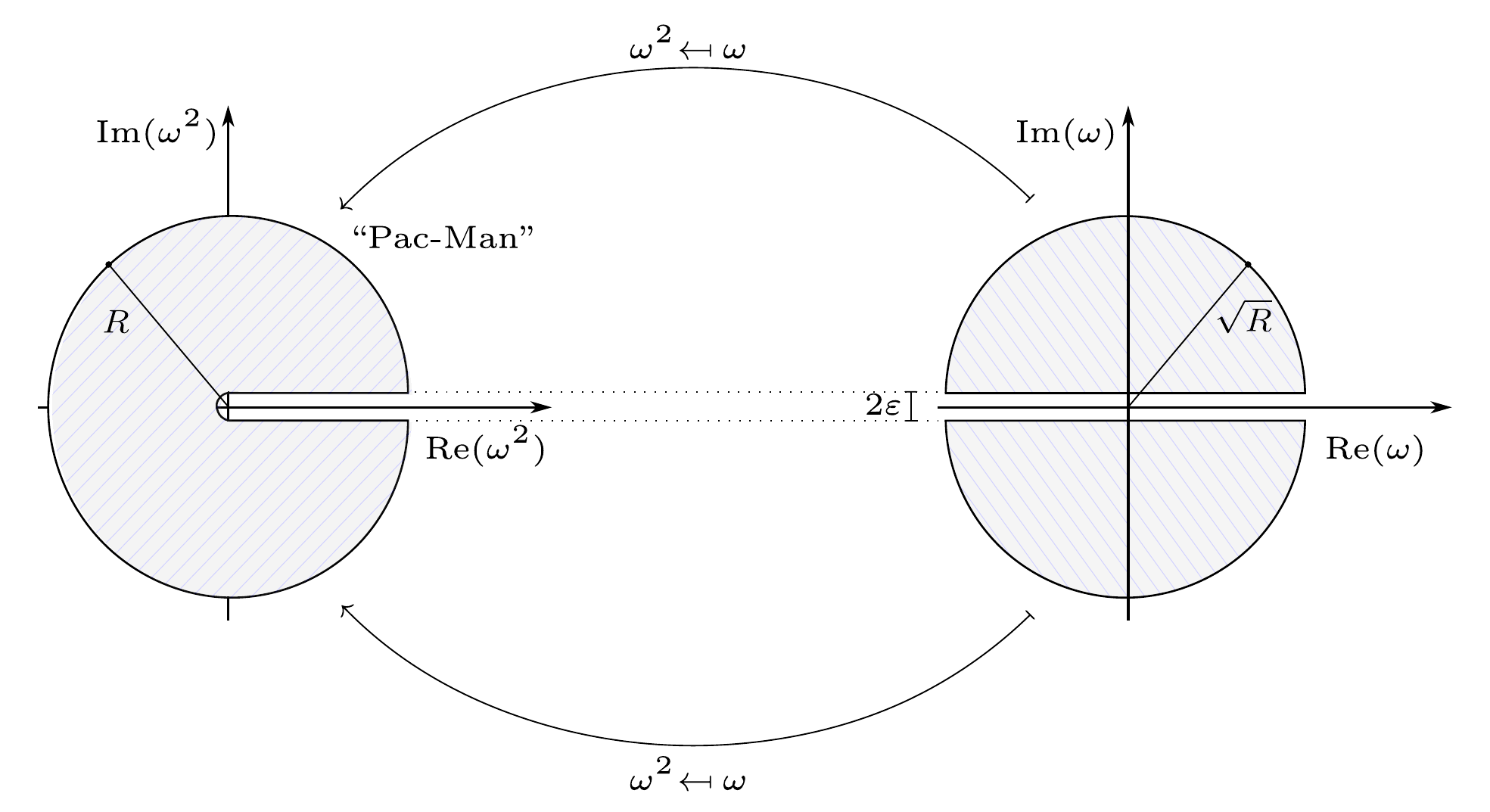}%
	\caption{Pac-Man and semi-disks contours.}
	\label{fig:contour}
\end{figure}

Next, let us outline the computation of the integral given in \Eref{eq: contour integral equal real-line integral}, to highlight the fact that all three possible choices of a secondary solution does yield the resolution of the identity, as per \Eref{eq:spectral resolution radial green function}. For $\mathcal{G}_1$, this computation is standard, \cite[Ex.7.3.2,Pg.454]{greenBook}. Yet, we do provide a step-by-step solution for all $\kappa$ in a supplementary notebook available online \cite{gitgit}.

Observe that the common part of the integrand coming from the Green functions in \Eref{eq: green functions exponential form harmonic oscillator explicitly all three k} yields
\begin{align*}
	-\frac{1}{2\pi i}\oint\limits_{\text{"pac-man"}}\left\{\frac{-i s}{2\omega} \left[e^{-i s \omega (x-x')} + e^{i s \omega (x+x')}\right]\right\} d\omega^2 
	& = \delta(x-x') + \delta(x+x').
\end{align*}
The remaining terms of $\mathcal{G}_1$ and $\mathcal{G}_3$ are obtained analogously to the above. For $\mathcal{G}_1$, we use the Fourier transform
\begin{align*}
	\frac{1}{2\pi}\int\limits_{-\infty}^{\infty} dy  e^{i ay}  \frac{1}{1\pm i y} = e^{\mp a} \Theta(\pm a),
\end{align*}
and the extra term gives a vanishing contribution. Since both $x$ and $x'$ are positive, we obtain \Eref{eq:spectral resolution radial green function} for each $\kappa$:
\begin{align*}
	&-\frac{1}{2\pi i}\oint\limits_{\text{"pac-man"}} \mathcal{G}_\kappa(x,x')d\omega^2  = \delta(x-x') + f_\kappa(\gamma)\delta(x+x')= \delta(x-x'),
\end{align*}
where we have introduced the auxiliary functions
\begin{align*}
	&f_1(\gamma) = 1,\\
	&f_2(\gamma)= - \cos(2\gamma),\\
	&f_3(\gamma) = -  \left[\frac{2 \cos (\gamma ) (2 \sin (\gamma )+\cos (\gamma ))}{2 \sin (2 \gamma )-\cos (2 \gamma )+ 3}  \right].
\end{align*}
With the spectral resolutions above, we can construct three one-parameter families of two-point functions, since it holds that
\begin{align*}
	-\frac{1}{2\pi i}\oint\limits_{\text{"pac-man"}} \mathcal{G}_\kappa(x,x')d\omega^2& = \int\limits_{-\infty}^{\infty} \omega d\omega  \frac{\psi_\kappa(x)\psi_\kappa(x')}{\mathcal{N}_\kappa},
\end{align*}
where the normalizations are given by
\begin{align*}
	\mathcal{N}_1&= \pi\omega[\cos^2(\gamma) +\omega^2 \sin^2(\gamma)],\\
	\mathcal{N}_2 &= \pi \omega \\
	\mathcal{N}_3 
	&= \pi \omega \frac{3 - \cos(2\gamma) + 2\sin(2\gamma)}{2}.
\end{align*}
\subsection{The two point functions}
Directly from the spectral resolution, as stated in Section \ref{sec: The quantum dynamics}, we obtain the spatial part of the two-point function in each case:
\begin{equation*}
	\widetilde{\psi}_\kappa(x,x') = \frac{\psi_\kappa(x)\psi_\kappa(x')}{\mathcal{N}_\kappa}.
\end{equation*}
With the three choices of secondary solutions we obtain three one-parameter families of two-point functions for three different ground states in the $2$-dimensional half-Minkowski spacetime:
\begin{align}
	\label{eq: ho 2point functions}
	\Psi_\kappa(t,x,t',x') = \int\limits_{0}^\infty d\omega e^{i \omega (t-t' - i 0^+)} \widetilde{\psi}_\kappa(x,x').
\end{align}

\begin{remark} Figure \ref{fig:ho} shows the behavior of $\widetilde{\psi}_\kappa(x,x')$ at coinciding points $x=x'=1$, for $\omega=-\Omega \in (0, 5)$, and for different values of $\gamma$. It manifests the inequivalence between the integrands of \Eref{eq: ho 2point functions} for different values of $\kappa$. One may ponder on the significance of such discrepancy, since it could be the case that different integrands yield equal integrals (as it happens for the resolution of the identity, for example). Yet, as it turns out, the term $\widetilde{\psi}_\kappa(x,x)$ has a physical interpretation: it characterizes the probability of de-excitations of a two-level system with energy gap $\Omega$, at a fixed spatial position $x$ and interacting for an infinite time with the quantum field in the ground state specified by $\Psi_\kappa(t,x,t',x')$. Such physical observable has been extensively used in the last years to probe a wide range of characteristics of the underlying quantum field theoretical framework, see \cite{DeSouzaCampos2022wsp,Louko2007mu} and the references therein.
	\begin{figure}[H]
		\centering
		\includegraphics[width=.48\textwidth]{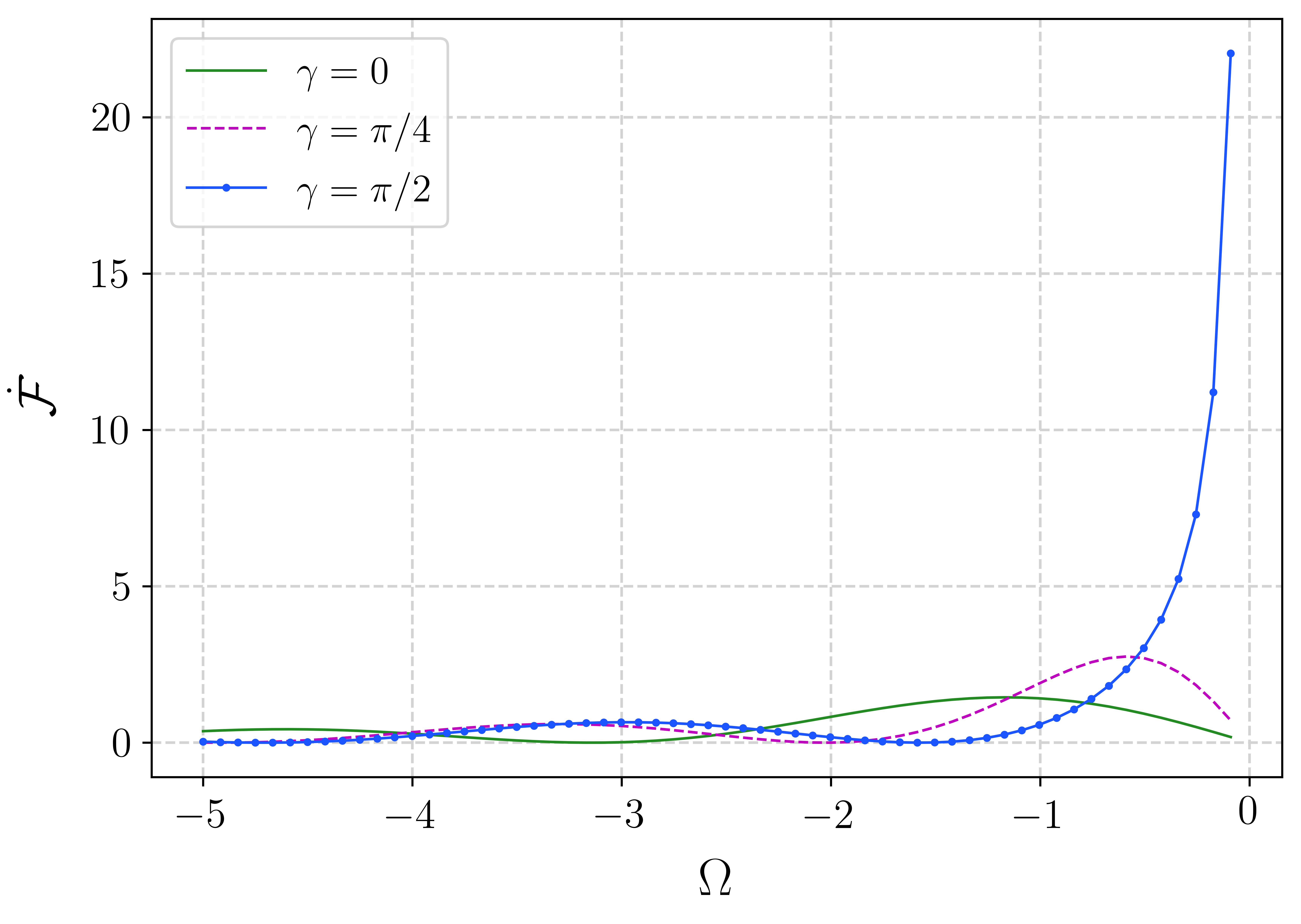}\hspace{.5cm}%
		\includegraphics[width=.48\textwidth]{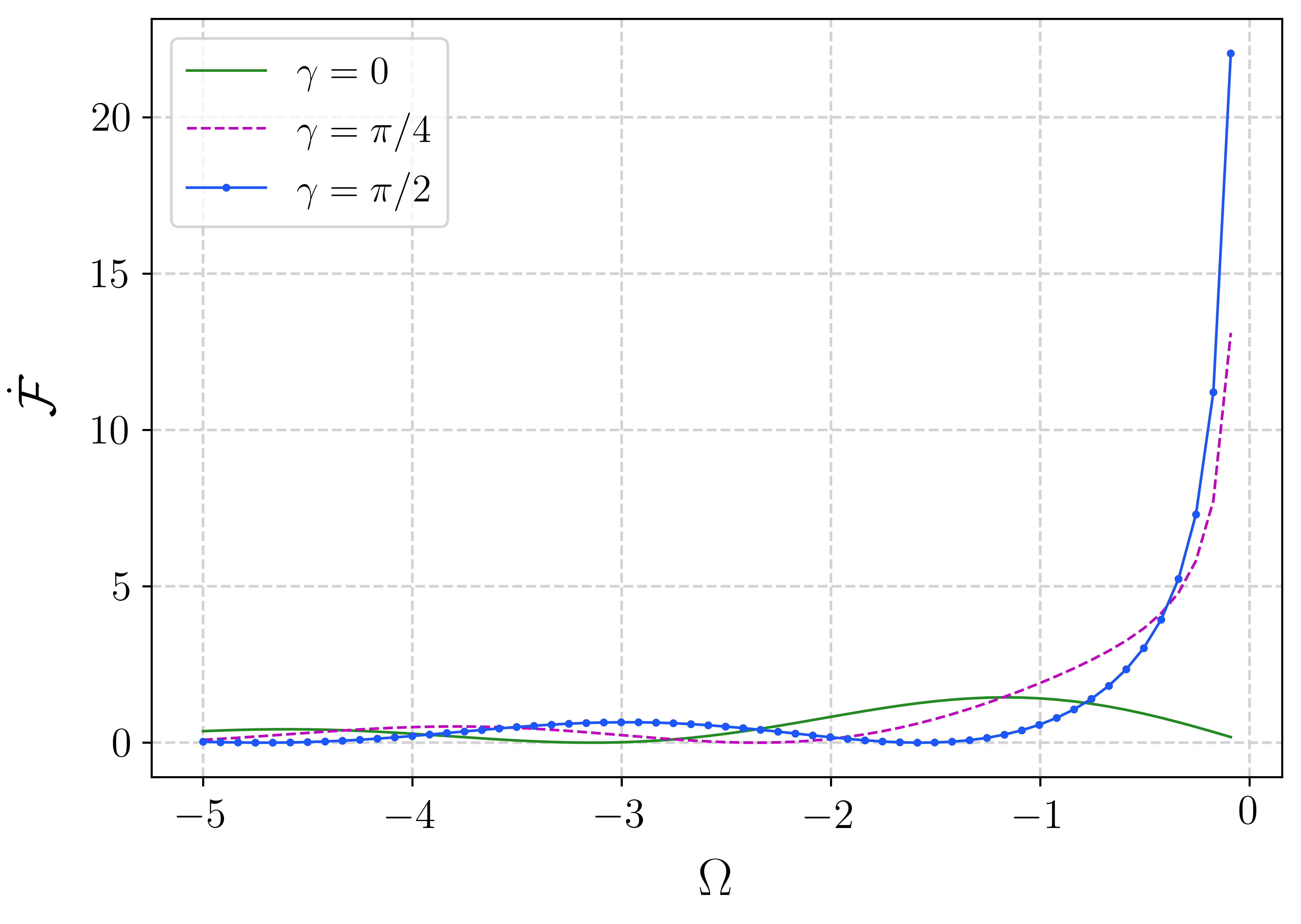}
		\caption{Behavior of $\dot{\mathcal{F}}=\widetilde{\psi}_\kappa(x,x)$ for $x=1$, $\omega = -\Omega$ and $\gamma\in\{0, \pi/4,\pi/2\}$. On the left, for $\kappa=1$. On the right, for $\kappa=2$.}
		\label{fig:ho}
	\end{figure}
\end{remark}
%
\section{Conclusion
	\label{sec: conclusions}}

In this work we have highlighted the existence of a hidden freedom in the standard procedure of constructing ground states for a real scalar field in a large class of static backgrounds. In particular we have observed that, when working at the level of the so-called radial equation, boundary conditions of Robin type can be imposed by using an arbitrary secondary solution. While this choice appears to be moot at the level of the underlying ordinary differential equation, it bears notable consequences at a fully covariant level. Interestingly, we have argued and shown via the concrete example of the two-dimensional half-Minkowski spacetime that, by exploring such freedom, one can account for a large class of boundary conditions that are structurally quite different from regular Robin boundary conditions, possibly including time-dependent ones.

From a structural viewpoint, the choice of secondary solution does not alter the effectiveness of the methods used until now for the construction of ground on static spacetimes. Nevertheless, it does open the possibility of studying a much larger class of boundary conditions and of investigating the physical consequence of the various different choices. To conclude we emphasize that a rationale similar to the one considered in this paper can be taken also in the investigation of boundary condition of Wentzell type, see {\it e.g.} \cite{Dappiaggi2022dwo,Dappiaggi2018pju}. Yet, a full-fledged analysis of this scenario would require a lengthy discussion that is worth leaving to a future work.

\ack
We are grateful for the discussions with professor J. Pitelli. The work of L.C. is supported by a postdoctoral fellowship of the Department of Physics of the University of Pavia, while that of L.S. by a PhD fellowship of the University of Pavia.

\section*{References}

\bibliographystyle{iopart-num}

\end{document}